\documentclass{aip-cp}

\usepackage[numbers]{natbib}
\usepackage{rotating}
\usepackage{graphicx}
\usepackage[T1]{fontenc}
\usepackage[latin9]{inputenc}
\usepackage{color}
\usepackage{units}
\usepackage{amssymb}

\usepackage{pgf,tikz}
\usepackage{mathrsfs}
\usetikzlibrary{arrows}
\definecolor{uuuuuu}{rgb}{0.26666666666666666,0.26666666666666666,0.26666666666666666}
\definecolor{qqwuqq}{rgb}{0.,0.39215686274509803,0.}

\makeatletter
\newtheorem{theorem}{\protect\theoremname}

\newtheorem{definition}[theorem]{\protect\definitionname}
\newtheorem{proposition}[theorem]{\protect\propositionname}
\newtheorem{example}[theorem]{\protect\examplename}
\newtheorem{conjecture}[theorem]{\protect\conjecturename}
\newtheorem{lemma}[theorem]{\protect\lemmaname}
\newenvironment{proof}[1][Proof]{\noindent\textbf{#1.} }{\ \rule{0.5em}{0.5em}}

\makeatother
  \providecommand{\corollaryname}{Corollary}
  \providecommand{\conjecturename}{Conjecture}
  \providecommand{\definitionname}{Definition}
  \providecommand{\examplename}{Example}
  \providecommand{\lemmaname}{Lemma}
  \providecommand{\propositionname}{Proposition}
  \providecommand{\theoremname}{Theorem}

\begin{document}

\title{Maximum Entropy and Sufficiency}
\author[aff1]{Peter Harremo{\"e}s\corref{cor1}}

\eaddress[url]{http://peter.harremoes.dk} 

\affil[aff1]{Niels Brock, Copenhagen Business College, Copenhagen, Denmark} 
\corresp[cor1]{Corresponding author: harremoes@ieee.org}

\maketitle

\begin{abstract}
The notion of Bregman divergence and sufficiency will be defined on general
convex state spaces. It is demonstrated that only spectral sets can have a
Bregman divergence that satisfies a sufficiency condition. Positive elements with trace 1 in a Jordan algebra are examples of spectral sets, and the most important example is the set of density matrices with complex entries. It is conjectured that information theoretic considerations lead directly to the notion of Jordan algebra under some regularity conditions.
\end{abstract}


\section{Introduction}

The maximum entropy method as introduced by Jaynes \cite{Jaynes1957} works quite well in various applications and an obvious question is why the same entropy formula appear in very different applications. In each of the applications the appearance of logarithmic terms have been given its own justification. Other scientists consider the maximum entropy principle as a general tool justified by information theory. It is quite obvious that the function that we would like to optimize should be convex or concave in order for our procedures to lead to a unique distribution. In this paper we will look at the consequences of requiring that our optimizer does not depend on irrelevant information, which is formulated as a sufficiency condition. 

The use of scoring rules has a long history in statistics. An early
contribution was the idea of minimizing the sum of square deviations
that dates back to Gauss and works perfectly for Gaussian distributions.
In the 1920's Ramsay and de Finetti proved versions of the Dutch book
theorem where determination of probability distributions were considered
as dual problems to maximizing a payoff function. Later it was proved
that any consistent inference corresponds to optimizing with respect
to some payoff function. A more systematic study of scoring rules
was given by McCarthy \cite{McCarthy1956}. The basic result is that the
only strictly local proper scoring rule is logarithmic score. Our main theorem extends this result to general regret functions on convex sets. 

Thermodynamics is the study of concepts like heat, temperature and
energy. A major objective is to extract as much energy from a system
as possible. Concepts like entropy and free energy play a significant
role. The idea in statistical mechanics is to view the macroscopic
behavior of a thermodynamic system as a statistical consequence of
the interaction between a lot of microscopic components where the
interacting between the components are governed by very simple laws.
Here the central limit theorem and large deviation theory play a major
role. One of the main achievements is the formula for entropy as a
logarithm of a probability.

One of the main purposes of information theory is to compress data
so that data can be recovered exactly or approximately. One of the
most important quantities was called entropy because it is calculated
according to a formula that mimics the calculation of entropy in statistical
mechanics. Another key concept in information theory is information
divergence (KL-divergence) that was introduced by Kullback and Leibler
in 1951 in a paper entitled information and sufficiency. The link
from information theory back to statistical physics was developed
by E.T. Jaynes via the maximum entropy principle. The link back to
statistics is now well established \cite{Barron1998,Csiszar2004}. 

The relation between information theory and gambling was established
by Kelly \cite{Kelly1956}. Logarithmic terms appear because we are
interested in the exponent in an exponential growth rate of of our
wealth. Later Kelly's approach has been generalized to traiding of
stocks although the relation to information theory is weaker \cite{Cover1991}. 

Since related quantities appear in statistics, statistical mechanics, information
theory and finance, and we are interested in a general theory that describes
when these relations are exact and when they just work by analogy. 
We introduce some general concepts related to optimization on convex
sets. These concepts apply exactly to all the topics under consideration
and lead to Bregman divergences. Then we introduce a notion of
sufficiency and show that this leads to information divergence and
logarithmic score. This second step is not always applicable which
explains when the different topics are really different. For applications in thermodynamics and gambling this is described in \cite{Harremoes2015a} and \cite{Harremoes2016a}. In this paper we will see how general convex optimization lead to the notion of Bregman divergence. If optimization is combined with the notion of sufficiency the Bregman divergence is generated by a function that is proportional to the entropy. This result holds on any convex set but it also gives very severe ties on the shape of the convex set to an extend that leads almost to the Hilbert space formalism of quantum mechanics. 

Due to the limited space in the proceedings paper most of the proofs have been foreshortened or omitted.

\section{Improved Caratheodory theorem}

We consider a situation where our knowledge about a system is given by an element in a convex set. These elements are called \emph{states} and convex combinations are formed by probabilistic mixing. States that cannot be distinguished by any measurement are considered as being the same state. The extreme points in the convex set are called \emph{pure states} and all other states are called \emph{mixed states}. See \cite{Holevo1982} for details about this definition of a state space. In this exposition we will assume that the state space, i.e. the convex set is finite dimensional and compact. 

If $C$ is a state space it is sometimes convenient to consider the positive cone generated by $C$. The positive cone consist of elements of the form $\lambda \cdot s$ where $\lambda \geq 0$ and $s\in C$. Elements of a positive cone can multiplied by positive constants via $\lambda\cdot\left(\mu\cdot s\right)=\left(\lambda\cdot\mu\right)\cdot s$ and can be added as follows.
\[
\lambda \cdot s_{1} + \mu \cdot s_{2}=\left(\lambda+\mu\right)\cdot\left( \frac{\lambda}{\lambda+\mu}s_{1}+ \frac{\mu}{\lambda+\mu}s_{2} \right).
\]
 The convex set and the positive cone can be embedded in a real vector space by taking the affine hull of the cone and use the apex of the cone as origin of the vector space. 

Let $x$ be an element in the positive cone such that
\begin{equation}
x=\sum_{i=1}^{n}\lambda_{i}\cdot s_{i}.  \label{eq:SpecDecomp}
\end{equation}
where $s_i$ are pure states. If $\lambda_{1}\geq \lambda_{2}\geq\dots\geq \lambda_{n}$ then
the vector $\lambda_{1}^{n}$ is called the \emph{spectrum of the decomposition} and such spectra are ordered by majorization. Let $\lambda_{1}^{n}$ and $\mu_{1}^{n}$  be the spectra of two decompositions of the same positive element. 
Then $\lambda_{1}^{n}\succeq \mu_{1}^{n}$  if $\sum_{i=1}^{k}\lambda_{i}\geq\sum_{i=1}^{k}\mu_{i}$ for $k\leq n$. Note that for positive element in a general positive cone the majorization ordering is a partial ordering. In special cases like the cone of positive semidefinite matrices on a complex Hilbert space the decompositions of the matrix form a lattice ordering with a unique maximal element, but in general the set of decompositions may have several incompatible maximal elements.   
Note that there is no restriction on the number $n$ in the definition of the spectrum, so if two spectra have different length we will extend the shorter vector by concatenating zeros at the end.

\begin{definition}
Let $C$ denote a convex set. A \emph{test} is an affine map from $C$ to $%
\left[ 0,1\right] .$ Let $s_{0}$ and $s_{1}$ denote states in the state
space $C.$ Then $s_{0}$ and $s_{1}$ are said to be \emph{mutually singular}
if there exists test $\phi $ such that $\phi \left( s_{0}\right) =0$ and $%
\phi \left( s_{1}\right) =1.$ The states $s_{0}$ and $s_{1}$ are said to be 
\emph{orthogonal} if $s_{0}$ and $s_{1}$ are mutually singular in the
smallest face $F$ of $C$ that contain both $s_{0}$ and $s_{1}.$ If the the extreme points $s_{1},s_{2},\dots,s_{n}$ of a decomposition are orthogonal then
the decomposition is called an \emph{%
orthogonal decomposition}.
\end{definition}

The trace of a positive element is defined by $\mathrm{Tr}\left(\lambda\cdot s \right)=\lambda$ when $s\in C$ so that states are positive elements with trace equal to 1. We note that the trace restricted to the states defines a test. Any tests can be identified with a positive functional on the positive cone that is dominated by the trace. We note that for a decomposition like (\ref{eq:SpecDecomp}) the trace is given by $\mathrm{Tr}\left(x\right)=\sum_{i=1}^{n} \lambda_{i}$. 

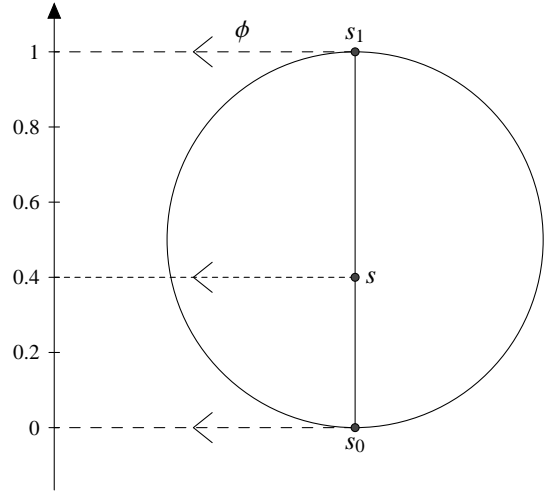
\begin{figure}
\centering
\begin{tikzpicture}[line cap=round,line join=round,>=triangle 45,x=5.0cm,y=5.0cm]
\draw[->,color=black] (0.,-0.16592592592592395) -- (0.,1.13185185185185);
\foreach \y in {0,0.2,0.4,0.6,0.8,1}
\draw[shift={(0,\y)},color=black] (2pt,0pt) -- (-2pt,0pt) node[left] {\footnotesize $\y$};

\clip(-0.17703703703703733,-0.16592592592592395) rectangle (1.363703703703703,1.13185185185185);
\draw(0.8,0.5) circle (2.5cm);
\draw [dash pattern=on 4pt off 4pt] (0.8,0.)-- (0.,0.);
\draw (0.37,0.) -- (0.42,0.04);
\draw (0.37,0.) -- (0.42,-0.04);
\draw [dash pattern=on 4pt off 4pt] (0.8,1.)-- (0.,1.);
\draw (0.37,1.) -- (0.42,1.04);
\draw (0.37,1.) -- (0.42,0.96);
\draw [dash pattern=on 2pt off 2pt] (0.8,0.4)-- (0.,0.4);
\draw (0.37,0.4) -- (0.42,0.44);
\draw (0.37,0.4) -- (0.42,0.36);
\draw (0.8,1.)-- (0.8,0.);
\draw (0.5,1) node[anchor=south] {$\phi$};
\draw (0.8,0) node[anchor=north] {$s_0$};
\draw (0.8,1) node[anchor=south] {$s_1$};
\draw (0.8,0.4) node[anchor= west] {$s$};
\begin{scriptsize}
\draw [fill=uuuuuu] (0.8,0.) circle (1.5pt);
\draw [fill=uuuuuu] (0.8,1.) circle (1.5pt);
\draw [fill=uuuuuu] (0.8,0.4) circle (1.5pt);
\end{scriptsize}
\end{tikzpicture}
\caption{In the disc the points $s_0$ and $s_1$ are mutually singular. The point $s$ has a unique decomposition into mutually singular points because it is not the center of the disc.}
\end{figure}

\begin{theorem}
Let $C$ denote a convex compact set of dimension $d$ and let $x$ denote
some element in the positive cone generated by $C$. Then there exists an orthogonal decomposition of the form as in Equation \ref{eq:SpecDecomp} such that $n\leq d+1$.
\end{theorem}

\begin{proof}
Without loss of generality we may assume that $x$ has trace 1 so that $x$ equals a state $s\in C$. We have to prove that there exists a probability vector $p_{i}^{n}$ and an orthogonal decomposition
\[
s=\sum_{i=1}^{n}p_{i}\cdot s_{i}.
\]
From the proof of Theorem 2 in \cite{Uhlmann1970} it follows that there exists a decomposition into at most $d+1$ extreme points that is maximal in the majorization ordering. Therefore assume that the decomposition is maximal. 
We will show that the states $s_{i}$ and $s_{j}$ are orthogonal.
Without loss of generality we may assume that $i=1$ and $j=2$. Now
\[
s  =\sum_{i=1}^{n}p_{i}\cdot s_{i}\\
  =\left(p_{1}+p_{2}\right)\left(\frac{p_{1}}{p_{1}+p_{2}}\cdot s_{0}+\frac{p_{2}}{p_{1}+p_{2}}\cdot s_{2}\right)+\sum_{i=3}^{n}p_{i}\cdot s_{i}\,.
\]
We will prove that $s_{0}$ and $s_{1}$ are singular in the smallest
face containing $s$. Without loss of generality we may assume that
$p_{0}+p_{1}=1$ and that $s$ is an algebraically interior point.

The proof is by induction on the dimension $d$. If $d=1$ the result is trivial.
Assume that the theorem has been proved for $d<n$ and that $C$ has
dimension $n$. Let $\tilde{C}$ be the intersection of $C$ and a
hyperplane through $s_{0}$ and $s_{1}.$ Then there exists a function
$\psi:\tilde{C}\to\left[0,1\right]$ such that $\psi\left(s_{i}\right)=i$.
Let $\ell$ denote the subset of the hyperplane where $\psi\left(x\right)=0.$
Then $C$ can be projected into a 2-dimensional vector space along
$\ell$. Therefore we just have to prove the result for $d=2.$

Introduce a coordinate system so that $s_{0}=\left(0,0\right)$ and
$s_{1}=\left(1,0\right)$ so that $s=\left(p_{1},0\right).$ Assume
$p_{1}<\nicefrac{1}{2}.$ Assume that $\left(0,1\right)$ is tangential
to $C$ in $s_{0}$ and that $\left(\alpha,-1\right)$ is tangential
to $C$ in $s_{1}$. Assume that $s=\left(1-p\right)\left(0,t_{1}\right)+p\left(1+t_{2}\alpha,t_{2}\left(-1\right)\right)$
where $t_{1},t_{2}>0.$ Then $p_{1}=p\left(1+\alpha t_{2}\right)$
so that $p\leq p_{1}$ if and only if $\alpha\leq0$.
\end{proof}

\begin{example}
In the unit square with $\left( 0,0\right) ,\left( 1,0\right) ,\left(
0,1\right) $ and $\left( 1,1\right) $ as vertices the point, with coordinates 
$\left( \nicefrac{1}{2},\nicefrac{1}{4}\right)$ an orthogonal decomposition with spectrum $\left( \nicefrac{1}{2},\nicefrac{1}{4},\nicefrac{1}{4}\right)$. This spectrum majorizes the spectrum of any other decomposition of this point, and it also majorizes the spectrum of any other point in the square. The square has in total four points symmetrically arranged with the same spectrum as $\left( \nicefrac{1}{2},\nicefrac{1}{4}\right)$.
\end{example}

Orthogonal decompositions are only unique when the convex set is a simplex. Nevertheless there exists a type of convex sets where a weaker type uniqueness holds.

\begin{definition}
If all orthogonal decompositions of a state have the same spectrum then the common
spectrum is called the \emph{spectrum of the state} and the state is said to
be \emph{spectral}. We say that the convex compact set $C$ is \emph{spectral}
if all states in $C$ are spectral. The rank of a spectral set is the maximal number of orthogonal states needed in an orthogonal decomposition of a state.
\end{definition}

A spectral set of entropic rank 2 is balanced, i.e. it is symmetric around a central point and
all boundary points are extreme. Two states in the set are orthogonal if and
only if they are antipodal. Any state can be decomposed into two antipodal
states. If the state is not the center of the balanced set this is the only
orthogonal decomposition. The center can be decomposed into a 
$\nicefrac{1}{2}$ and $\nicefrac{1}{2}$ mixture of any pair of antipodal
points.

\begin{proposition}
In two dimensions a simplex and a balanced set are the only spectral sets.
\end{proposition}

\section{Local Bregman divergences}

The entropy of an element $x$ of the positive cone can be defined as 
\[
H\left(x\right)=\inf\left(-\sum_{i=1}^{n}\lambda_{i}\ln\left(\lambda_{i}\right)\right)
\]
where the infimum is taken over all spectra of $x$. Since entropy is decreasing under majorization the entropy of $x$ is attained at an orthogonal decomposition. This definition extends a similar definition of the entropy of a state as defined by Uhlmann \cite{Uhlmann1970}. In general this definition of entropy does not give a concave function on the convex set. For instance the entropy of points in a square has four local maxima corresponding to the four points with maximal spectrum in the majorization ordering.

Let $\mathcal{A}$ denote a subset of the feasible measurements such
that $a\in \mathcal{A}$ maps $C$ into a distribution on the real numbers i.e.
the distribution of a random variable. The elements of $\mathcal{A}$ may represent feasible
\emph{actions} (decisions) that lead to a payoff like the score of
a statistical decision, the energy extracted by a certain interaction
with the system, (minus) the length of a codeword of the next encoded
input letter using a specific code book, or the revenue of using a
certain portfolio. For each $s\in C$ we define
\[
\left\langle a,s\right\rangle =E\left[a\left(s\right)\right].
\]
and
\[
F\left(s\right)=\sup_{a\in\mathcal{A}}\left\langle a,s\right\rangle .
\]
Without loss of generality we may assume that the set of actions $\mathcal{A}$
is closed so that we may assume that there exists $a\in\mathcal{A}$
such that $F\left(s\right)=\left\langle a,s\right\rangle $ and in
this case we say that $a$ is optimal for $s.$ We note that $F$
is convex but $F$ need not be strictly convex. 

\begin{definition}
If $F\left(s\right)$ is finite then we define \emph{the regret} of the action $a$
by 
\[
D_{F}\left(s,a\right)=F\left(s\right)-\left\langle a,s\right\rangle .
\]
\end{definition}

\begin{figure}
\centering
\begin{tikzpicture}[line cap=round,line join=round,>=triangle 45,x=10.0cm,y=10.0cm]
\draw[->,color=black] (0.,0.) -- (1.1419478737997253,0.);
\draw[->,color=black] (0.,-0.5) -- (0.,0.15580246913580265);
\clip(-0.06431,-0.5008779149519884) rectangle (1.1419478737997253,0.15580246913580265);
\draw[line width=1.2pt,color=qqwuqq,smooth,samples=100,domain=0:1.1419478737997253] plot(\x,{(\x)*ln((\x)+1.0E-4)});
\draw (0.3868, -0.4994) -- (1.1362, 0.0827);
\draw [dash pattern=on 4pt off 4pt] (0.8,0.)-- (0.8,-0.178414847300847);
\draw [line width=2.pt] (0.45,-0.450314607074793)-- (0.45,-0.35922847440746253);
\draw [dash pattern=on 4pt off 4pt] (0.45,0.)-- (0.45,-0.35922847440746253);
\draw (0.4343468211923174,0.07152263374485619) node[anchor=north west] {$s_0$};
\draw (0.790781098312178,0.0680109739369001) node[anchor=north west] {$s_1$};
\draw (0.27983343997779647,-0.38) node[anchor=north west] {$D_{F}(s_0 ,s_1)$};
\draw (1.0330861733985859,0.1540466392318246) node[anchor=north west] {$F$};
\begin{scriptsize}
\draw [fill=uuuuuu] (0.8,-0.178414847300847) circle (1.5pt);
\draw [fill=uuuuuu] (0.45,-0.35922847440746253) circle (1.5pt);
\draw [fill=uuuuuu] (0.45,-0.450314607074793) circle (1.5pt);
\draw [fill=uuuuuu] (0.45,0.) circle (1.5pt);
\draw [fill=uuuuuu] (0.8,0.) circle (1.5pt);
\end{scriptsize}
\end{tikzpicture}
\label{breg}
\caption{The regret equals the vertical distance  between the curve and the tangent.}
\end{figure}
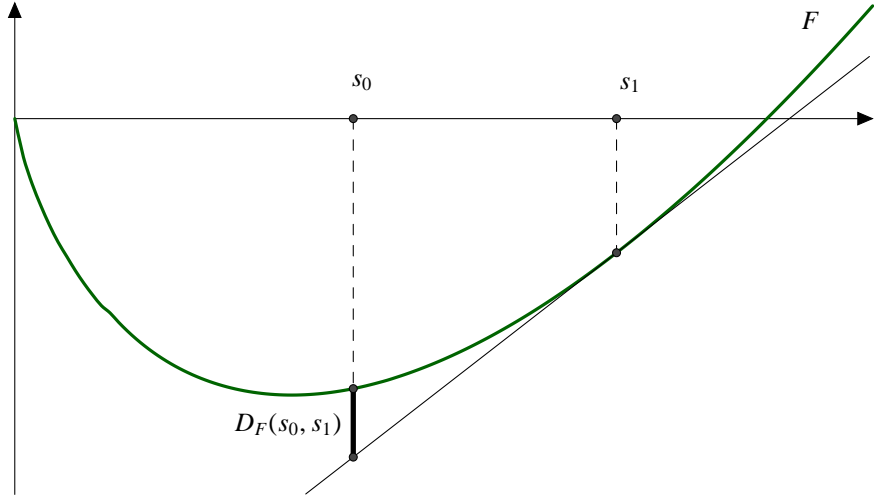

If the state is $s_{1}$ but one acts as if the state were $s_{2}$
one suffers a regret that equals the difference between what one achieves
and what could have been achieved. 
\begin{definition}\label{def:regret}
If $F\left(s_{1}\right)$ is finite then we define \emph{the regret of the state
$s_{2}$} as 
\[
D_{F}\left(s_{1},s_{2}\right)=\inf_{a}D_{F}\left(s_{1},a\right)
\]
where the infimum is taken over actions $a$ that are optimal for
$s_{2}.$ 
\end{definition}

The notion of sufficiency for Bregman divergences have been introduced in \cite{Harremoes2007a} and \cite{Jiao2014}. It was shown in \cite{Jiao2014} that a Bregman divergence on the simplex of distributions on an alphabet that is not binary determines the divergence except for a multiplicative factor.

\begin{definition}
Let $C$ denote a convex set and let $\Phi :C\rightarrow C$ denote
some affine map. Then $\Phi $ $is$ said to be sufficient for the family of
states $s_{\theta }$ if there exists an affine transformation $\Psi
:C\rightarrow C$ such that $\Psi \left( \Phi \left( s_{\theta }\right)
\right) =s_{\theta }.$ Let $D_F$ denote a regret function defined according to Definition \ref{def:regret}. Then $D_F$ is said
to satisfy sufficiency if 
\[
D_{F}\left( \Phi \left( s_{1}\right) \Vert \Phi \left( s_{2}\right) \right)
=D_{F}\left( s_{1}\Vert s_{2}\right) 
\]%
for any states $s_{1},s_{2}\in C$ and any affine transformation $\Phi
:C\rightarrow C$ that is sufficient for $s_{1},s_{2}.$
\end{definition}

Recently it has been proved that divergence on a complex Hilbert space is
decreasing under positive trace preserving maps \cite{Mueller-Hermes2015, Christandl2016}. Therefore information divergence satisfies the
sufficiency condition on complex Hilbert spaces. Hence sufficiency is also
satisfied on real Hilbert spaces. It it not known if sufficiency holds on
more general Jordan algebras so we introduce a weaker condition called
locality.

\begin{definition}
The regret function $D_{F}$ is said to be local if 
\[
D_{F}\left( \left( 1-t\right) s_{0}+ts_{1}\Vert s_{0}\right) =D_{F}\left( \left(
1-t\right) s_{0}+ts_{2}\Vert s_{0}\right) 
\]%
when $s_{1}$ and $s_{2}$ are states that are orthogonal to $s_{0}.$ 
\end{definition}

\begin{proposition}
Let $C$ denote a spectral convex set. Then the Bregman divergence generated
by the entropy is local.
\end{proposition}

\begin{proof}
Assume that $s=\left(1-p\right)s_{0}+ps_{1}$ where $s_{0}$ and $s_{1}$
are orthogonal. Then one can make orthogonal decompositions 
\[
s_{0}  =\sum p_{0i}\cdot s_{0i}\, \mathrm{ and } \, s_{1}  =\sum p_{1j}\cdot s_{1j} .
\]
Then
\[
D_{H}\left(s_{0},s\right)  = \sum p_{0i}\cdot\ln\frac{p_{0i}}{\left(1-p\right)p_{0i}} = \sum p_{0i}\cdot\ln\frac{1}{1-p} =\ln\frac{1}{1-p},
\]
which does not depend on $s_{1}$ as long as $s_{1}$ is orthogonal
to $s_{0}$.\end{proof}

\begin{proposition}
If the regret function $D_{F}$ on a convex set satisfies the sufficiency
condition then it is local.
\end{proposition}

\begin{proposition}
Let $C$ denote a spectral convex set of entropic rank 2. Then the convex set
is balanced and any Bregman divergence is local.
\end{proposition}

The following lemma follows from Alexandrov's theorem. See \cite{Rockafeller1970} Theorem 25.5 for details.

\begin{lemma}
A convex function on a finite dimensional convex set is differentiable
almost everywhere with respect to the Lebesgue measure.
\end{lemma}

\begin{theorem}
Let $C$ be a convex set with at least three orthogonal states. If a
regret function $D_{F}$ defined on $C$ is local then it is a Bregman
divergence generated by the entropy times some constant.
\end{theorem}

Note that the constant mentioned in the theorem has to be negative in order to get a convex function. Obviously one could have defined Bregman divergences of concave functions in almost the same way as we have defined Bregman divergences for convex functions.

\begin{proof}
Let $K$ denote the convex hull of a set $s_{0},s_{1},\dots s_{n}$
of orthogonal states. Let $f_{i}$ denote the function $f_{i}\left(x\right)=D_{F}\left(s_{i},xs_{i}+\left(1-x\right)s_{i+1}\right)$.
Note that $f_{i}$ is continuous. Let $P=\sum p_{i}s_{i}$ and $Q=\sum q_{i}P_{i}.$
If $F$ is differentiable in $P$ then locality implies that
\[
D_{F}\left(P,Q\right)  = \sum p_{i}d\left(s_{i},Q\right)-\sum p_{i}d\left(s_{i},P\right) = \sum p_{i}f_{i}\left(q_{i}\right)-\sum p_{i}f_{i}\left(p_{i}\right)
  = \sum p_{i}\left(f_{i}\left(q_{i}\right)-f_{i}\left(p_{i}\right)\right).
\]
Since $F$ is differentiable almost surely on $K$ we see that continuity
implies that the equality most hold for all distributions $P.$

As a function of $Q$ the regret has minimum when $Q=P.$ We have 
\[
x\left(f_{i}\left(y\right)-f_{i}\left(x\right)\right)+z\left(f_{j}\left(w\right)-f_{j}\left(z\right)\right)\geq0.
\]
where $x+z=y+w.$ We also have 
\[
x\left(f_{j}\left(y\right)-f_{j}\left(x\right)\right)+z\left(f_{i}\left(w\right)-f_{i}\left(z\right)\right)\geq0
\]
implying that 
\[
x\left(f_{ij}\left(y\right)-f_{ij}\left(x\right)\right)+z\left(f_{ij}\left(w\right)-f_{ij}\left(z\right)\right)\geq0
\]
 where $f_{ij}=\frac{f_{i}+f_{j}}{2}.$ 

Assume that $x=z=\frac{y+w}{2}$. Then
\[
\frac{f_{ij}\left(y\right)+f_{ij}\left(w\right)}{2}  \geq  f_{ij}\left(x\right)
\]
so that $f_{ij}$ is convex. Therefore $f_{ij}$ is differentiable
from left and right. We have
\[
\left(y+\epsilon\right)\left(f_{ij}\left(y\right)-f_{ij}\left(y+\epsilon\right)\right)+\left(y-\epsilon\right)\left(f_{ij}\left(w\right)-f_{ij}\left(y-\epsilon\right)\right)\geq0
\]
with equality when $\epsilon=0.$ We differentiate with respect to
$\epsilon$ from right.
\[
\left(f_{ij}\left(y\right)-f_{ij}\left(y+\epsilon\right)\right)+\left(y+\epsilon\right)\left(-f_{ij+}'\left(y+\epsilon\right)\right)-\left(f_{ij}\left(w\right)-f_{ij}\left(y-\epsilon\right)\right)+\left(y-\epsilon\right)\left(f_{ij-}'\left(y-\epsilon\right)\right)
\]
which is positive for $\epsilon=0.$ This implies that
\[
y\cdot f_{ij-}'\left(y\right) \geq y\cdot f_{ij+}'\left(y\right).
\]
Since $f_{ij}$ is convex we have $f_{ij-}'\left(y\right)\leq f_{ij+}'\left(y\right)$
which in combination with the previous inequality implies that $f_{ij-}'\left(y\right)=f_{ij+}'\left(y\right)$
so that $f_{ij}$ is differentiable. Since $f_{i}=f_{ij}+f_{ik}-f_{jk}$
the function $f_{i}$ is also differentiable. 

Since $f_i$ is differentiable the regret function is a Bregman divergence, we can use Thm. 4 in \cite{Jiao2014} to conclude that there exists an affine function defined on $K$ such that
$F_{\mid K}=-c_{K}\cdot H_{\mid K}+g_{K}$. If $K$ and $L$ simplices
such that $x\in K\cap L$ then 
\[
-c_{K}\cdot H_{\mid K}\left(x\right)+g_{K}\left(x\right)=-c_{L}\cdot H_{\mid L}\left(x\right)+g_{L}\left(x\right)
\]
so that 
\[
\left(c_{L}-c_{K}\right)\cdot H_{\mid K}\left(x\right)=g_{L}\left(x\right)-g_{K}\left(x\right).
\]
If $K\cap L$ has dimension greater than zero then the right hand
side is affine so the left hand side is affine which is only possible
when $c_{K}=c_{L}.$ Therefore we also have $g_{L}\left(x\right)=g_{K}\left(x\right)$
for all $x\in K\cap L.$ Therefore the functions $g_{K}$ can be extended
to function on the whole of $C.$
\end{proof}

The previous proof sketch can even be used to demonstrate that a convex set with local Bregman divergence must be spectral. The notion of a spectral set is related to self-duality of the cone of positive elements, which leads to the following conjecture.
\begin{conjecture}
If a finite dimensional convex compact set has a local regret function and has a transitive symmetry
group then the convex set can be represented as positive elements of a
Jordan algebra with trace 1.
\end{conjecture}

\section{Concavity of entropy in Jordan algebras}

The density matrices with complex entries play a crucial role in the
mathematical theory of quantum mechanics and it is well-known that the set of
density matrices is a spectral set. For each density matrix the spectrum
equals the usual spectrum calculated as roots of the characteristic
polynomial. In the 1930'ties Jordan generalized the notion of Hermitian
complex matrix to the notion of a Jordan algebra in an attempt to provide an
alternative to the complex Hilbert spaces as the mathematical basis of
quantum mechanics. For instance the complex
Hermitean matrices form a Jordan algebra with the composition $x\circ y=%
\frac{1}{2}\left( xy+yx\right) .$ 

An Euclidean Jordan algebra is an algebra with composition $\circ $ that is
commutative and satisfies the Jordan identity 
\[
\left( x\circ y\right) \circ \left( x\circ x\right) =x\circ \left( y\circ
\left( x\circ x\right) \right) .
\]%
Further it is assumed that 
\[
\sum_{i=1}^{n}x_{i}^{2}=0
\]
implies that $x_{i}=0$ for all $i.$ In an Euclidean Jordan algebra we write $%
x\geq 0$ if $x$ is a sum of squares. 

For a number in a real division algebra $tr$ is defined as the real part
of the number so that $tr\left( xy\right) =tr\left( yx\right) $. For a matrix $%
\left( M_{mn}\right) $ the trace $\mathrm{Tr}$ is defined by $\mathrm{Tr}%
\left( M\right) =tr\left( \sum_{n}M_{nn}\right) .$ Then $\mathrm{Tr}\left( MN\right)  =\mathrm{Tr}\left( NM\right)$ .

In a finite dimensional Eucledian Jordan algebra the density operators are defined as the positive elements with trace 1. With this definition the set of
density operators of a Jordan algebra is a spectral set. Any Euclidean Jordan algebra can be decomposed into
5 types leading to the following convex sets:

\begin{description}
\item[Real] Density matrices over the real numbers.

\item[Complex] Density matrices over the complex numbers.

\item[Quaternionic] Density matrices over the quaternions.

\item[Exceptional] $3\times 3$ density matrices with entries that are
octonions.

\item[Spin type] A unit ball in $d$ dimensions.
\end{description}
See \cite{McCrimmon2004} for general results on Jordan algebras and \cite {Adler1995} for details about quantum mechanics described using quaternions.

The entropy is defined as for general convex set and we will prove that $H$
is a concave on the cone of positive elements. For the unit ball the entropy is centrally
symmetric and is obviously concave. The following exposition is based on a
similar result for complex matrix algebras stated in \cite{Petz2008}, but
the proofs have been changed so that they are valid without assuming commutativity or associativity
of the division algebra.

\begin{lemma}
Let $A$ and $B$ denote Hermitian matrices. Assume that $A=\sum t_{\ell
}E_{\ell }$ where $E_{\ell }$ are orthogonal idempotents. If $f$ is a
holomorphic function around the spectrum of $A$ then 
\[
\frac{\mathrm{d}}{\mathrm{d}t}f\left( A+tB\right) _{\mid
t=0}=\sum_{m,n}a_{mn}\cdot E_{m}BE_{n}.
\]%
where 
\[
a_{mn}=\frac{f\left( t_{m}\right) -f\left( t_{n}\right) }{t_{m}-t_{n}}  \,\mathrm{for} \,t_{m}\neq t_{n}
\]
and $a_{mn}=f^{\prime }\left( t_{m}\right)$ for $t_{m}=t_{n}.$%
\end{lemma}

\begin{proof}
First assume that $f\left(z\right)=z^{r}.$ Then the result follows directly from an expansion of $\left(A+tB\right)^{r}$ followed by differentiation. Therefore the theorem also holds for all polynomials and  thereby also for any holomorphic function because such functions can be approximated by polynomials.
\end{proof}

By taking the trace on each side of the equation in the previous lemma we get the following result.

\begin{lemma}
For Hermitian matrices $A$ and $B$ we have 
\[
\left. \frac{\mathrm{d}}{\mathrm{d}t}\mathrm{Tr}\left( f\left( A+tB\right)
\right) \right\vert _{t=0}=\mathrm{Tr}\left( f^{\prime }\left( A\right)
B\right) .
\]
\end{lemma}

\begin{theorem}
In a Jordan algebra the entropy is a strictly
concave function on the cone of positive elements.
\end{theorem}

\begin{proof}
Let $f$ denote the holomorphic function $f\left( z\right) =-z\ln z,~z>0.$
We have to prove that $\mathrm{Tr}\left( f\left( \left( 1-t\right)
A+tX\right) \right) =\mathrm{Tr}\left( f\left( A+tB\right) \right) $ is
concave where $B=X-A$.  The second derivative can be calculated. 
\[
\left. \frac{\mathrm{d}^{2}}{\mathrm{d}t^{2}}\mathrm{Tr}\left( f\left(
A+tB\right) \right) \right\vert _{t=0} =\left. \frac{\mathrm{d}}{\mathrm{d}%
t}\mathrm{Tr}\left( f^{\prime }\left( A+tB\right) B\right) \right\vert _{t=0}
 =\mathrm{Tr}\left( \left. \frac{\mathrm{d}}{\mathrm{d}t}f^{\prime }\left(
A+tB\right) \right\vert _{t=0}B\right)= \mathrm{Tr}\left( \left( \sum_{m,n}\tilde{a}_{mn}\cdot E_{m}BE_{n}\right)
B\right) 
\]
where 
\[
\tilde{a}_{mn}=
\frac{f^{\prime }\left( t_{m}\right) -f^{\prime }\left( t_{n}\right) }{t_{m}-t_{n}} \, \mathrm{for}\,t_{m}\neq t_{n} 
\]
and $\tilde{a}_{mn}=f^{\prime \prime }\left( t_{m}\right)$ for $t_{m}=t_{n}.$ Then
\[
\left. \frac{\mathrm{d}^{2}}{\mathrm{d}t^{2}}\mathrm{Tr}\left( f\left(
A+tB\right) \right) \right\vert _{t=0}
= \sum_{m,n}\tilde{a}_{mn}\cdot \mathrm{Tr}\left( E_{m}BE_{n}B\right)  
= \sum_{m,n}\tilde{a}_{mn}\cdot \mathrm{Tr}\left( \left( E_{m}BE_{n}\right)
\left( E_{m}BE_{n}\right) ^{\ast }\right) 
\]
Since $\tilde{a}_{mn} < 0$ and $\mathrm{Tr}\left( \left(
E_{m}BE_{n}\right) \left( E_{m}BE_{n}\right) ^{\ast }\right) < 0$ we have 
$\frac{\mathrm{d}^{2}}{\mathrm{d}t^{2}}\mathrm{Tr}f\left( A+tB\right) \leq 0.
$
\end{proof}


%


\begin{thebibliography}{18}%
\makeatletter
\providecommand \@ifxundefined [1]{%
 \@ifx{#1\undefined}
}%
\providecommand \@ifnum [1]{%
 \ifnum #1\expandafter \@firstoftwo
 \else \expandafter \@secondoftwo
 \fi
}%
\providecommand \@ifx [1]{%
 \ifx #1\expandafter \@firstoftwo
 \else \expandafter \@secondoftwo
 \fi
}%
\providecommand \natexlab [1]{#1}%
\providecommand \enquote  [1]{``#1''}%
\providecommand \bibnamefont  [1]{#1}%
\providecommand \bibfnamefont [1]{#1}%
\providecommand \citenamefont [1]{#1}%
\providecommand \href@noop [0]{\@secondoftwo}%
\providecommand \href [0]{\begingroup \@sanitize@url \@href}%
\providecommand \@href[1]{\@@startlink{#1}\@@href}%
\providecommand \@@href[1]{\endgroup#1\@@endlink}%
\providecommand \@sanitize@url [0]{\catcode `\$12\catcode `\&12\catcode
  `\#12\catcode `\^12\catcode `\_12\catcode `\%12\relax}%
\providecommand \@@startlink[1]{}%
\providecommand \@@endlink[0]{}%
\providecommand \url  [0]{\begingroup\@sanitize@url \@url }%
\providecommand \@url [1]{\endgroup\@href {#1}{\urlprefix }}%
\providecommand \urlprefix  [0]{URL }%
\providecommand \Eprint [0]{\href }%
\providecommand \doibase [0]{http://dx.doi.org/}%
\providecommand \selectlanguage [0]{\@gobble}%
\providecommand \bibinfo  [0]{\@secondoftwo}%
\providecommand \bibfield  [0]{\@secondoftwo}%
\providecommand \translation [1]{[#1]}%
\providecommand \BibitemOpen [0]{}%
\providecommand \bibitemStop [0]{}%
\providecommand \bibitemNoStop [0]{.\EOS\space}%
\providecommand \EOS [0]{\spacefactor3000\relax}%
\providecommand \BibitemShut  [1]{\csname bibitem#1\endcsname}%
\let\auto@bib@innerbib\@empty
\bibitem [{\citenamefont {Jaynes}(1957)}]{Jaynes1957}%
  \BibitemOpen
  \bibfield  {author} {\bibinfo {author} {\bibfnamefont {E.~T.}\ \bibnamefont
  {Jaynes}},\ }\href@noop {} {\bibfield  {journal} {\bibinfo  {journal}
  {Physical Reviews}\ }\textbf {\bibinfo {volume} {106 and 108}},\ \unskip\
  \bibinfo {pages} {620--630 and 171--190} (\bibinfo {year}
  {1957})}\BibitemShut {NoStop}%
\bibitem [{\citenamefont {{McC}arthy}(1956)}]{McCarthy1956}%
  \BibitemOpen
  \bibfield  {author} {\bibinfo {author} {\bibfnamefont {J.}~\bibnamefont
  {{McC}arthy}},\ }\href@noop {} {\bibfield  {journal} {\bibinfo  {journal}
  {Proc. Nat. Acad. Sci.}\ }\textbf {\bibinfo {volume} {42}},\ \unskip\
  \bibinfo {pages} {654--655} (\bibinfo {year} {1956})}\BibitemShut {NoStop}%
\bibitem [{\citenamefont {Barron}, \citenamefont {Rissanen},\ and\
  \citenamefont {Yu}(1998)}]{Barron1998}%
  \BibitemOpen
  \bibfield  {author} {\bibinfo {author} {\bibfnamefont {A.~R.}\ \bibnamefont
  {Barron}}, \bibinfo {author} {\bibfnamefont {J.}~\bibnamefont {Rissanen}}, \
  and\ \bibinfo {author} {\bibfnamefont {B.}~\bibnamefont {Yu}},\ }\href@noop
  {} {\bibfield  {journal} {\bibinfo  {journal} {IEEE Trans. Inform. Theory}\
  }\textbf {\bibinfo {volume} {44}},\ \unskip\ \bibinfo {pages}
  {2743--2760}Oct. (\bibinfo {year} {1998})},\ \bibinfo {note} {commemorative
  issue}\BibitemShut {NoStop}%
\bibitem [{\citenamefont {Csisz{\'a}r}\ and\ \citenamefont
  {Shields}(2004)}]{Csiszar2004}%
  \BibitemOpen
  \bibfield  {author} {\bibinfo {author} {\bibfnamefont {I.}~\bibnamefont
  {Csisz{\'a}r}}\ and\ \bibinfo {author} {\bibfnamefont {P.}~\bibnamefont
  {Shields}},\ }\href@noop {} {\emph {\bibinfo {title} {Information Theory and
  Statistics: A Tutorial}}},\ Foundations and Trends in Communications and
  Information Theory\ (\bibinfo  {publisher} {Now Publishers Inc.},\ \bibinfo
  {year} {2004})\BibitemShut {NoStop}%
\bibitem [{\citenamefont {Kelly}(1956)}]{Kelly1956}%
  \BibitemOpen
  \bibfield  {author} {\bibinfo {author} {\bibfnamefont {J.~L.}\ \bibnamefont
  {Kelly}},\ }\href@noop {} {\bibfield  {journal} {\bibinfo  {journal} {Bell
  System Technical Journal}\ }\textbf {\bibinfo {volume} {35}},\ \unskip\
  \bibinfo {pages} {917--926} (\bibinfo {year} {1956})}\BibitemShut {NoStop}%
\bibitem [{\citenamefont {Cover}\ and\ \citenamefont
  {Thomas}(1991)}]{Cover1991}%
  \BibitemOpen
  \bibfield  {author} {\bibinfo {author} {\bibfnamefont {T.}~\bibnamefont
  {Cover}}\ and\ \bibinfo {author} {\bibfnamefont {J.~A.}\ \bibnamefont
  {Thomas}},\ }\href@noop {} {\emph {\bibinfo {title} {Elements of Information
  Theory}}}\ (\bibinfo  {publisher} {Wiley},\ \bibinfo {year}
  {1991})\BibitemShut {NoStop}%
\bibitem [{\citenamefont {Harremo{\"e}s}(2015)}]{Harremoes2015a}%
  \BibitemOpen
  \bibfield  {author} {\bibinfo {author} {\bibfnamefont {P.}~\bibnamefont
  {Harremo{\"e}s}},\ }\enquote {\bibinfo {title} {Proper scoring and
  sufficiency},}\ in\ \href
  {http://www.cs.helsinki.fi/u/ttonteri/pub/witmse2015proceedings.pdf
  http://arxiv.org/abs/1507.07089} {\emph {\bibinfo {booktitle} {Proceeding of
  the The Eighth Workshop on Information Theoretic Methods in Science and
  Engineering}}},\ \bibinfo {series and number} {\bibinfo {series} {Series of
  Publications B}\ \ \bibinfo {number} {Report B-2015-1}},\ \bibinfo {editor}
  {edited by\ \bibinfo {editor} {\bibfnamefont {J.}~\bibnamefont {Rissanen}},
  \bibinfo {editor} {\bibfnamefont {P.}~\bibnamefont {Harremo{\"e}s}}, \bibinfo
  {editor} {\bibfnamefont {S.}~\bibnamefont {Forchhammer}}, \bibinfo {editor}
  {\bibfnamefont {T.}~\bibnamefont {Roos}}, \ and\ \bibinfo {editor}
  {\bibfnamefont {P.}~\bibnamefont {Myllym{\"a}ke}}}\ (\bibinfo {address}
  {University of Helsinki, Department of Computer Science},\ \bibinfo {year}
  {2015})\ \unskip, pp.\ \bibinfo {pages} {19--22},\ \bibinfo {note} {an
  appendix with proofs only exists in the arXiv version of the
  paper}\BibitemShut {NoStop}%
\bibitem [{\citenamefont {Harremo{\"e}s}(2016)}]{Harremoes2016a}%
  \BibitemOpen
  \bibfield  {author} {\bibinfo {author} {\bibfnamefont {P.}~\bibnamefont
  {Harremo{\"e}s}},\ }\href@noop {} {\enquote {\bibinfo {title} {Sufficiency on
  the stock market},}\ }Jan. (\bibinfo {year} {2016}),\ \bibinfo {note}
  {arXiv:1601.07593}\BibitemShut {NoStop}%
\bibitem [{\citenamefont {Holevo}(1982)}]{Holevo1982}%
  \BibitemOpen
  \bibfield  {author} {\bibinfo {author} {\bibfnamefont {A.~S.}\ \bibnamefont
  {Holevo}},\ }\href@noop {} {\emph {\bibinfo {title} {Probabilistic and
  Statistical Aspects of Quantum Theory}}},\ edited by\ \bibinfo {editor}
  {\bibfnamefont {P.~R.}\ \bibnamefont {Krishnaiah}}, \bibinfo {editor}
  {\bibfnamefont {C.~R.}\ \bibnamefont {Rao}}, \bibinfo {editor} {\bibfnamefont
  {M.}~\bibnamefont {Rosenblatt}}, \ and\ \bibinfo {editor} {\bibfnamefont
  {Y.~A.}\ \bibnamefont {Rozanov}},\ \bibinfo {series} {North-Holland Series in
  Statistics and Probability}, Vol.~\bibinfo {volume} {1}\ (\bibinfo
  {publisher} {North-Holland},\ \bibinfo {address} {Amsterdam},\ \bibinfo
  {year} {1982})\BibitemShut {NoStop}%
\bibitem [{\citenamefont {Uhlmann}(1970)}]{Uhlmann1970}%
  \BibitemOpen
  \bibfield  {author} {\bibinfo {author} {\bibfnamefont {A.}~\bibnamefont
  {Uhlmann}},\ }\href@noop {} {\bibfield  {journal} {\bibinfo  {journal}
  {Reports on Mathematical Physics}\ }\textbf {\bibinfo {volume} {1}},\
  \unskip\ \bibinfo {pages} {147--159} (\bibinfo {year} {1970})}\BibitemShut
  {NoStop}%
\bibitem [{\citenamefont {Harremo{\"e}s}\ and\ \citenamefont
  {Tishby}(2007)}]{Harremoes2007a}%
  \BibitemOpen
  \bibfield  {author} {\bibinfo {author} {\bibfnamefont {P.}~\bibnamefont
  {Harremo{\"e}s}}\ and\ \bibinfo {author} {\bibfnamefont {N.}~\bibnamefont
  {Tishby}},\ }\enquote {\bibinfo {title} {The information bottleneck revisited
  or how to choose a good distortion measure},}\ in\ \href
  {www.harremoes.dk/Peter/flaske2.pdf} {\emph {\bibinfo {booktitle}
  {Proceedings ISIT 2007, Nice}}}\ (\bibinfo {organization} {IEEE Information
  Theory Society},\ \bibinfo {year} {2007})\ \unskip, pp.\ \bibinfo {pages}
  {566--571}\BibitemShut {NoStop}%
\bibitem [{\citenamefont {Jiao}\ \emph {et~al.}(2014)\citenamefont {Jiao},
  \citenamefont {amd Albert~No}, \citenamefont {Venkat},\ and\ \citenamefont
  {Weissman}}]{Jiao2014}%
  \BibitemOpen
  \bibfield  {author} {\bibinfo {author} {\bibfnamefont {J.}~\bibnamefont
  {Jiao}}, \bibinfo {author} {\bibfnamefont {T.~C.}\ \bibnamefont {amd
  Albert~No}}, \bibinfo {author} {\bibfnamefont {K.}~\bibnamefont {Venkat}}, \
  and\ \bibinfo {author} {\bibfnamefont {T.}~\bibnamefont {Weissman}},\
  }\href@noop {} {\bibfield  {journal} {\bibinfo  {journal} {Trans. Inform.
  Theory}\ }\textbf {\bibinfo {volume} {60}},\ \unskip\ \bibinfo {pages}
  {7616--7626}Dec. (\bibinfo {year} {2014})}\BibitemShut {NoStop}%
\bibitem [{\citenamefont {M{\"u}ller-Hermes}\ and\ \citenamefont
  {Reeb}(2015)}]{Mueller-Hermes2015}%
  \BibitemOpen
  \bibfield  {author} {\bibinfo {author} {\bibfnamefont {A.}~\bibnamefont
  {M{\"u}ller-Hermes}}\ and\ \bibinfo {author} {\bibfnamefont {D.}~\bibnamefont
  {Reeb}},\ }\href@noop {} {\enquote {\bibinfo {title} {Monotonicity of the
  quantum relative entropy under positive maps},}\ } (\bibinfo {year} {2015}),\
  \bibinfo {note} {arXiv:1512.06117}\BibitemShut {NoStop}%
\bibitem [{\citenamefont {Christandl}\ and\ \citenamefont
  {M{\"u}ller-Hermes}(2016)}]{Christandl2016}%
  \BibitemOpen
  \bibfield  {author} {\bibinfo {author} {\bibfnamefont {M.}~\bibnamefont
  {Christandl}}\ and\ \bibinfo {author} {\bibfnamefont {A.}~\bibnamefont
  {M{\"u}ller-Hermes}},\ }\href@noop {} {\enquote {\bibinfo {title} {Relative
  entropy bounds on quantum, private and repeater capacities},}\ } (\bibinfo
  {year} {2016}),\ \bibinfo {note} {arXiv:1604.03448}\BibitemShut {NoStop}%
\bibitem [{\citenamefont {Rockafellar}(1970)}]{Rockafeller1970}%
  \BibitemOpen
  \bibfield  {author} {\bibinfo {author} {\bibfnamefont {R.~T.}\ \bibnamefont
  {Rockafellar}},\ }\href@noop {} {\emph {\bibinfo {title} {Convex Analysis}}}\
  (\bibinfo  {publisher} {Princeton Univ. Press},\ \bibinfo {address} {New
  Jersey},\ \bibinfo {year} {1970})\BibitemShut {NoStop}%
\bibitem [{\citenamefont {{McC}rimmon}(2004)}]{McCrimmon2004}%
  \BibitemOpen
  \bibfield  {author} {\bibinfo {author} {\bibfnamefont {K.}~\bibnamefont
  {{McC}rimmon}},\ }\href@noop {} {\emph {\bibinfo {title} {A Taste of Jordan
  Algebras}}}\ (\bibinfo  {publisher} {Springer},\ \bibinfo {year}
  {2004})\BibitemShut {NoStop}%
\bibitem [{\citenamefont {Adler}(1995)}]{Adler1995}%
  \BibitemOpen
  \bibfield  {author} {\bibinfo {author} {\bibfnamefont {S.~L.}\ \bibnamefont
  {Adler}},\ }\href@noop {} {\emph {\bibinfo {title} {Quaternionic Quantum
  Mechanics and Quantum Fields}}}\ (\bibinfo  {publisher} {Oxford Univ.
  Press},\ \bibinfo {address} {New York, Oxford},\ \bibinfo {year}
  {1995})\BibitemShut {NoStop}%
\bibitem [{\citenamefont {Petz}(2008)}]{Petz2008}%
  \BibitemOpen
  \bibfield  {author} {\bibinfo {author} {\bibfnamefont {D.}~\bibnamefont
  {Petz}},\ }\href@noop {} {\emph {\bibinfo {title} {Quantum information theory
  and quantum statistics}}}\ (\bibinfo  {publisher} {Springer},\ \bibinfo
  {year} {2008})\BibitemShut {NoStop}%
\end{thebibliography}

\end{document}